\documentclass[runningheads]{llncs}
\usepackage{amssymb}
\usepackage{amsmath}
\usepackage{comment}
\usepackage[disable]{todonotes}
\usepackage[utf8]{inputenc}
\usepackage{shuffle}
\usepackage{multirow}
\usepackage{mathabx}
\usepackage[sectionbib, square,sort,comma,numbers]{natbib}
\usepackage{hyperref}

\usepackage{diagbox}
\usepackage{slashbox}
\usepackage{tikz}
\usetikzlibrary{arrows,shapes,automata,positioning}

\usepackage{apxproof}
\theoremstyle{plain}
\newtheoremrep{theorem}{Theorem}%[section]
 \newtheoremrep{proposition}[theorem]{Proposition}
\newtheoremrep{lemma}[theorem]{Lemma}
 \newtheoremrep{claim}[theorem]{Claim}
 \newtheoremrep{conjecture}[theorem]{Conjecture}
 \newtheoremrep{corollary}[theorem]{Corollary}
 \newtheoremrep{definition}[theorem]{Definition}

\newcommand{\perm}{\operatorname{perm}}

\DeclareMathOperator{\lcm}{lcm}

\newenvironment{claiminproof}[1]{\medskip\par\noindent\underline{Claim:}\space#1}{}
\newenvironment{claimproof}[1]{\begin{quote}\par\noindent\emph{Proof of the Claim:}\space#1}{[\emph{End, Proof of the Claim}]\end{quote}}%\newline}
\def\dotminus{\mathbin{\ooalign{\hss\raise1ex\hbox{.}\hss\cr
  \mathsurround=0pt$-$}}} 
 
\makeatletter
\providecommand*{\shuffle}{%
  \mathbin{\mathpalette\shuffle@{}}%
}
\newcommand*{\shuffle@}[2]{%
  \sbox0{$#1\vcenter{}$}%
  \kern .15\ht0 % side bearing
  \rlap{\vrule height .25\ht0 depth 0pt width 2.5\ht0}%
  \raise.1\ht0\hbox to 2.5\ht0{%
    \vrule height 1.75\ht0 depth -.1\ht0 width .17\ht0 %
    \hfill
    \vrule height 1.75\ht0 depth -.1\ht0 width .17\ht0 %
    \hfill
    \vrule height 1.75\ht0 depth -.1\ht0 width .17\ht0 %
  }%
  \kern .15\ht0 % side bearing
}
\makeatother

\newcommand{\itshuffle}{\shuffle,*}
\newcommand{\plusshuffle}{\shuffle,+}

\DeclareFontFamily{U}{bigshuffle}{}
\DeclareFontShape{U}{bigshuffle}{m}{n}{
  <5-8> s*[1.7] shuffle7
  <8->  s*[1.7] shuffle10
}{}
\DeclareSymbolFont{BigShuffle}{U}{bigshuffle}{m}{n}
\DeclareMathSymbol\bigshuffle{\mathop}{BigShuffle}{"001}
\DeclareMathSymbol\bigcshuffle{\mathop}{BigShuffle}{"002}

\newcommand{\NP}{\textsf{NP}}

\begin{document}
%
%\title{State Complexity of Projection on Commutative Languages}
%\title{State Complexity of Certain Operations and Closure Properties on Commutative Languages}
 % oder state complexity of projection and shuffle and closure properties ...
\title{Regularity Conditions for Iterated Shuffle on Commutative Regular Languages}
\titlerunning{Regularity of Iterated Shuffle on Commutative Regular Languages}

\author{Stefan Hoffmann\orcidID{0000-0002-7866-075X}}
\authorrunning{S. Hoffmann}
\institute{Informatikwissenschaften, FB IV, 
  Universit\"at Trier,  Universitätsring 15, 54296~Trier, Germany, 
  \email{hoffmanns@informatik.uni-trier.de}}
\maketitle              % typeset the header of the contribution
\begin{abstract}
 We identify a subclass
 of the regular commutative languages that is closed
 under the iterated shuffle, or shuffle closure.
 In particular, it is regularity-preserving on this subclass. This subclass contains 
 the commutative group languages
 and, for every alphabet $\Sigma$, the class $\textbf{Com}^+(\Sigma^*)$
 given by the ordered variety $\textbf{Com}^+$.
 Then, we state a simple characterization
 when the iterated shuffle on finite commutative languages
 gives a regular language again and state partial results for aperiodic
 commutative languages.
 We also show that the  aperiodic, or star-free, commutative languages
 and the commutative group languages are closed under projection.

\end{abstract}

\section{Introduction}

% ergebnisse finite + allgemein. + unäre

% therem closure comm languages

The shuffle and iterated shuffle have been introduced and studied to understand, or specify,
the semantics of parallel programs. This was undertaken, as it appears
to be, independently by Campbell and Habermann~\cite{CamHab74}, by Mazurkiewicz~\cite{DBLP:conf/mfcs/Marzurkiewicz75}
and by Shaw~\cite{Shaw78zbMATH03592960}. They introduced \emph{flow expressions}, 
which allow for sequential operators (catenation and iterated catenation) as well
as for parallel operators (shuffle and iterated shuffle)
to specify sequential and parallel execution traces.
%Further studies
%of the resulting language classes where carried out in~\cite{DBLP:journals/tcs/Jantzen85,FerPSV2017}.

For illustration, let us reproduce the following very simple Reader-Writer Problem from~\cite{Shaw78zbMATH03592960}, as an example involving the iterated shuffle. In this problem, a set of cyclic
processes may be in read-mode, but only one process at a time is allowed to be in write-mode, and read and write operations may not proceed concurrently.
Additionally, we impose that the processes have to come to an end, in~\cite{Shaw78zbMATH03592960} they are allowed
to run indefinitely.
This constraint could be specified, using our notation, by
\[
 ((\operatorname{StartRead\cdot Read \cdot EndRead})^{\itshuffle} \cup \operatorname{Write})^{*},
\]
where ``$\itshuffle$'' denotes the iterated shuffle
and ``$*$'' the Kleene star.

Let us note that in~\cite{Shaw78zbMATH03592960} additional lock
and signal instructions were allowed. Also in~\cite{DBLP:journals/cl/Riddle79} similar expressions for process modeling were investigated, allowing the binary shuffle operation, but without inclusion of the iterated shuffle.

The shuffle operation as a binary operation, but not the iterated shuffle,
is regularity-preserving on all regular languages. 
However, already the iterated shuffle of very simple languages
can give non-regular languages.
Hence, it is interesting to know, and to identify, 
quite rich classes for which this operation
is regularity-preserving. Here, we give such
a class which includes the commutative group 
languages and the languages described by the
positive variety $\mathbf{Com}^+$.
Additionally, we give a characterization for the regularity
of the iterated shuffle when applied to finite 
commutative languages and state some partial results for aperiodic (or star-free) commutative languages.

We mention that subregular language classes closed under the binary shuffle operation
were investigated previously~\cite{DBLP:journals/tcs/Perrot78,DBLP:journals/actaC/AlmeidaEP17,DBLP:journals/fuin/CastiglioneR12,DBLP:conf/lata/Restivo15,DBLP:journals/iandc/BerstelBCPR10,DBLP:journals/tcs/GomezP04}.

We also show that the commutative
star-free languages and the 
commutative group languages
are closed under projections.
For further connections on regularity conditions
and closure properties, in particular for the star-free
languages, see the recent survey~\cite{DBLP:conf/lata/Pin20}.

\section{Preliminaries and Definitions}
\label{sec::preliminaries}

\subsection{General Notions} 

% eher so kurz halten wie in masopast projection paper...?

Let $\Sigma$ be a finite set of symbols
 called an \emph{alphabet}. The set $\Sigma^{\ast}$ denotes
the set of all finite sequences, i.e., of all \emph{words}. The finite sequence of length zero,
or the \emph{empty word}, is denoted by $\varepsilon$. For a given word we denote by $|w|$
its length, and for $a \in \Sigma$ by $|w|_a$ the number of occurrences of the symbol $a$
in $w$. A \emph{language} is a subset of $\Sigma^*$.
If $L \subseteq \Sigma^*$ and $u \in \Sigma^*$,
then the \emph{quotients} are the languages $u^{-1}L =  \{ v \in \Sigma^* \mid uv \in L \}$ and $Lu^{-1} = \{ v \in \Sigma^* \mid vu \in L \}$.

We assume the reader
to have some basic knowledge in formal language
theory, as contained, e.g., in~\cite{HopUll79,DBLP:books/daglib/0088160}. For instance, we make use of regular expressions to describe
languages.

%\begin{definition}
%\label{def:pi_Sigmaj}
 Let $\Gamma \subseteq \Sigma$. Then, we define
 \emph{projection homomorphisms} $\pi_{\Gamma} : \Sigma^* \to \Gamma^*$ onto $\Gamma^*$
 by
%  \[
%   \pi_{\Gamma}(x) = \left\{
%   \begin{array}{ll}
%   x           & \mbox{if } x \in \Gamma; \\ 
%   \varepsilon & \mbox{otherwise.}
%   \end{array}
%   \right.
%  \]
 $\pi_{\Gamma}(x) = x$ for $x \in \Gamma$
 and $\pi_{\Gamma}(x) = \varepsilon$
 for $x \notin \Gamma$.
 
%\end{definition}

By $\mathbb N_0 = \{ 0,1,2,\ldots \}$,
we denote the set of natural numbers, including zero. 
We will also consider the ordered set $\mathbb N_0 \cup \{\infty\}$ with $\mathbb N_0$
having the usual order and setting $n < \infty$ for any $n \in \mathbb N_0$.
%By $\mathcal P(X)$, we denote the \emph{power set} of a set~$X$.

A quintuple $\mathcal A = (\Sigma, Q, \delta, q_0, F)$
is a finite \emph{(incomplete) deterministic automaton}, where
 $\delta : Q \times \Sigma \to S$ is a partial transition function, $Q$ a finite set of states, $q_0 \in S$
the start state and $F \subseteq Q$ the set of final states. 
The automaton $\mathcal A$ is said to be \emph{complete} if $\delta$
is a total function.
The transition function $\delta : Q \times \Sigma \to S$
could be extended to a transition function on words $\delta^{\ast} : Q \times \Sigma^{\ast} \to S$
by setting $\delta^{\ast}(q, \varepsilon) = q$ and $\delta^{\ast}(q, wa) := \delta(\delta^{\ast}(q, w), a)$
for $q \in Q$, $a \in \Sigma$ and $w \in \Sigma^{\ast}$. In the remainder, we drop
the distinction between both functions and will also denote this extension by $\delta$.
%If $\delta$ is a partial function, the automaton is called \emph{incomplete}.
%A \emph{nondeterministic automaton} is given by a function $\delta : Q \times \Sigma^* \to \mathcal P(Q)$,
%with all the other, related, notions appropriate modified, see~\cite{Ito04}.
%Most of the time, the automata considered in this paper will be complete, deterministic and initially connected, the last notion meaning
%for every $s \in S$ there exists some $w \in \Sigma^{\ast}$ such that $\delta(s_0, w) = s$.
%A \emph{trap} (\emph{sink}) \emph{state} is a state $s \in Q$
%such that, for any $x \in \Sigma$, $\delta(s, x) = s$.
The language \emph{recognized} by an automaton $\mathcal A = (\Sigma, Q, \delta, q_0, F)$ is
$
 L(\mathcal A) = \{ w \in \Sigma^{\ast} \mid \delta(q_0, w) \in F \}.
$
A language $L \subseteq \Sigma^{\ast}$ is called \emph{regular} if $L = L(\mathcal A)$
for some finite automaton~$\mathcal A$.

\begin{toappendix}
The \emph{Nerode right-congruence} 
with respect to $L \subseteq \Sigma^*$ is defined, for $u,v \in \Sigma^*$, by $u \equiv_L v$ if and only if 
\[
 \forall x \in \Sigma^* : ux \in L \leftrightarrow vx \in L.
\]
The equivalence class, for $w \in \Sigma^{\ast}$,
is denoted by $[w]_{\equiv_L} = \{ x \in \Sigma^{\ast} \mid x \equiv_L w \}$.
A language is regular if and only if the above right-congruence has finite index.
\end{toappendix}

The following classic result will also be needed later.

\begin{theorem}[Generalized Chinese Remainder Theorem~\cite{schmid58}]
\label{thm:CRT}
 The system of linear congruences
 \[
  x \equiv r_i \pmod{m_i} \quad (i=1,2,\ldots,k)
 \] 
 has integral solutions $x$ if and only if $\gcd(m_i,m_j)$ divides $(r_i - r_j)$
 for all pairs $i \ne j$ and all solutions are congruent modulo $\lcm(m_1, \ldots, m_k)$.
\end{theorem}

\subsection{Commutative Languages and the Shuffle Operation}
\label{subsec:com_lang_shuffle}

% Let $\Sigma = \{a_1, \ldots, a_k\}$ be an alphabet. The map $\psi : \Sigma^{\ast} \to \mathbb N^k$
% given by $\psi(w) = (|w|_{a_1}, \ldots, |w|_{a_k})$ is called the \emph{Parikh morphism}.
% If $L \subseteq \Sigma^*$, we set $\psi(L) = \{ \psi(w) \mid w \in L \}$.
% For a given word $w \in \Sigma^{\ast}$, we define $\perm(w) := \{ u \in \Sigma^{\ast} : \psi(u) = \psi(w) \}$.
% If $L \subseteq \Sigma^{\ast}$, then we set $\perm(L) := \bigcup_{w\in L} \perm(w)$.
% A language is called \emph{commutative},
% if $\perm(L) = L$.

For a given word $w \in \Sigma^{\ast}$, we define $\perm(w) := \{ u \in \Sigma^{\ast} \mid \forall a \in \Sigma : |u|_a = |w|_a \}$.
If $L \subseteq \Sigma^{\ast}$, then we set $\perm(L) := \bigcup_{w\in L} \perm(w)$.
A language is called \emph{commutative},
if $\perm(L) = L$. Let $\Sigma = \{a_1, \ldots, a_k\}$.
The \emph{Parikh mapping} is $\psi : \Sigma^* \to \mathbb N_0^k$
given by $\psi(u) = (|u|_{a_1}, \ldots, |u|_{a_k})$ for $u \in \Sigma^*$.
We have $\perm(L) = \psi^{-1}(\psi(L))$.

%\begin{definition} 
The \emph{shuffle operation}, denoted by $\shuffle$, is defined by
%
%  \begin{align*}
%     u \shuffle v & := \left\{ \begin{array}{ll}
%      \multirow{2}{*}{$x_1 y_1 x_2 y_2 \cdots x_n y_n  \mid$} &  u = x_1 x_2 \cdots x_n, v = y_1 y_2 \cdots y_n, \\ 
%          &   x_i, y_i \in \Sigma^{\ast}, 1 \le i \le n, n \ge 1
%   \end{array} \right\},
%  \end{align*}
 \begin{multline*}%\label{def:shuffle}
    u \shuffle v  = \{ w \in \Sigma^*  \mid  w = x_1 y_1 x_2 y_2 \cdots x_n y_n 
    \mbox{ for some words } \\ x_1, \ldots, x_n, y_1, \ldots, y_n \in \Sigma^*
    \mbox{ such that } u = x_1 x_2 \cdots x_n \mbox{ and } v = y_1 y_2 \cdots y_n \},
 \end{multline*}
 for $u,v \in \Sigma^{\ast}$ and 
  $L_1 \shuffle L_2  := \bigcup_{x \in L_1, y \in L_2} (x \shuffle y)$ for $L_1, L_2 \subseteq \Sigma^{\ast}$.
%\end{definition}

 In writing formulas
 without brackets, we suppose that the shuffle operation binds stronger than the set operations,
 and the concatenation operator has the strongest binding. 
 
  If $L_1, \ldots, L_n \subseteq \Sigma^*$, we set $\bigshuffle_{i=1}^n L_i = L_1 \shuffle \ldots \shuffle L_n$.
 The \emph{iterated shuffle} of $L \subseteq \Sigma^*$ 
 is $L^{\shuffle,*} = \bigcup_{n \ge 0} \bigshuffle_{i=1}^n L$.
 We also set $L^{\plusshuffle} = \bigcup_{n \ge 1} \bigshuffle_{i=1}^n L$.
 
 \begin{theorem}[Fernau et al. \cite{FerPSV2017}] % todo + andere referenzen aus jantzen und co
 \label{thm:shuffle_properties}
  Let $U,V,W \subseteq \Sigma^*$. Then,
  \begin{enumerate}
  \item $U \shuffle V = V \shuffle U$ (commutative law);
  \item $(U \shuffle V) \shuffle W = U \shuffle (V \shuffle W)$ (associative law);
  \item $U \shuffle (V \cup W)  
   = (U \shuffle V) \cup (U \shuffle W)$ (distributive over union);
  \item $(U^{\shuffle,*})^{\shuffle,*} = U^{\shuffle,*}$;
  \item $(U\cup V)^{\shuffle,*} = U^{\shuffle,*} \shuffle V^{\shuffle,*}$;
  \item $(U \shuffle V^{\shuffle,*})^{\shuffle,*} = (U \shuffle (U \cup V)^{\shuffle,*}) \cup \{\varepsilon\}$.
  \end{enumerate}
 \end{theorem}

The next
result is taken from~\cite{FerPSV2017} and gives 
equations like $\perm(UV) = \perm(U) \shuffle \perm(V)$
or $\perm(U^*) = \perm(U)^{\shuffle,*}$ for $U,V \subseteq \Sigma^*$.
A \emph{semiring} is an algebraic structure $(S, +, \cdot, 0, 1)$
such that $(S, +, 0)$ forms a commutative monoid, $(S, \cdot, 1)$ is a monoid
and we have $a\cdot (b + c) = a\cdot b + a\cdot c$, $(b+c)\cdot a = b\cdot a + c\cdot a$
and $0 \cdot a = a \cdot 0 = 0$.

\begin{theorem}[Fernau et al.  \cite{FerPSV2017}]
\label{thm:perm_semiring_hom}
 $\perm : \mathcal P(\Sigma^*) \to \mathcal P(\Sigma^*)$
 is a semiring morphism from the semiring
 $(\mathcal P(\Sigma^*), \cup, \cdot, \emptyset, \{\varepsilon\})$, that also respects the iterated catenation resp. iterated shuffle operation,
 to the semiring $(\mathcal P(\Sigma^*), \cup, \shuffle, \emptyset, \{\varepsilon\})$.
\end{theorem}

 The class of commutative languages obeys the following closure properties.

\begin{theoremrep}[\cite{Pin86,DBLP:reference/hfl/Pin97,DBLP:conf/cai/Hoffmann19,Hoffmann20}]
\label{thm::comm_lang_closure} 
 The class of commutative languages is closed under union, intersection, complement,
 projections, the shuffle operation and the iterated shuffle.
\end{theoremrep}
\begin{proof}
 In \cite{Pin86}, closure under the boolean operations and shuffle is shown.
 For closure under projections, note that
 for any $L \subseteq \Sigma^*$ and $\Gamma \subseteq \Sigma$,
 we have $\perm(\pi_{\Gamma}(L)) = \pi_{\Gamma}(\perm(L))$.
 For iterated shuffle, if $u \in L^{\itshuffle}$,
 then $u \in L^{\shuffle,n}$ for some $n \ge 0$.
 Hence, as $L^{\shuffle,n}$ is commutative, $\perm(u) \subseteq L^{\shuffle,n}$.
 So, $L^{\itshuffle}$ is also commutative.~\qed
 \end{proof}

\subsection{Aperiodic and Group Languages}
\label{subsec:aperiodic_group}

The class of aperiodic
languages was introduced in~\cite{McNaughton71} and admits a wealth of other
characterizations.

\begin{definition}
\label{def:aperiodic_aut}
 An automaton $\mathcal A = (\Sigma, Q, \delta, q_0, F)$
 is \emph{aperiodic}, if there exists $n \ge 0$
 such that, for all states $q \in Q$ and any
 word $w \in \Sigma^*$, we have
 $
  \delta(q, w^n) = \delta(q, w^{n+1}).
 $
\end{definition}

We define the class of aperiodic languages.

\begin{definition}
\label{def:aperiodic_lang}
 A regular language is called \emph{aperiodic} if
 there exists an aperiodic automaton recognizing it.
\end{definition}

The class of \emph{star-free regular languages} is the smallest class
containing $\{\varepsilon\}, \Sigma^*$ and $\{a\}$ for any $a \in \Sigma$
and closed under the boolean operations and concatenation.
Let us state the following, due to~\cite{DBLP:journals/iandc/Schutzenberger65a}.

\begin{theorem}[Schützenberger~\cite{DBLP:journals/iandc/Schutzenberger65a,McNaughton71}]
\label{thm:star-free_equals_aperiodic}
 The class of star-free languages equals the class 
 of aperiodic languages.
\end{theorem}

Next, we introduce the group languages.

\begin{definition}[McNaughton~\cite{McNaughton67}]
\label{def:grp_lang}
A \emph{(pure-)group language}\footnote{These were
introduced in~\cite{McNaughton67} under the name of pure-group events.}
is a language recognized by an automaton $\mathcal A = (\Sigma, Q, \delta, q_0, F)$
where every letter acts as a permutation on the state set\footnote{
Such automata are also called \emph{permutation automata}, and the name stems
from the fact that the transformation monoid of such an automaton forms a group.}, %definitionen trafo monoid, group todo. todo, def synt. monoid.
%tood sprechweise acts on, oder permutes the set of states
i.e., if $a \in \Sigma$, then the map $\delta_a : Q \to Q$
given by $\delta_a(q) = \delta(a,q)$ for $q \in Q$
is total and a permutation of $Q$. Such an automaton is called a \emph{permutation automaton}.
\end{definition}

Observe that a permutation automaton, as defined here, is always complete\footnote{Another
way would be, to allow incomplete automata, to insist that every letter either gives a permutation
or labels no transition.}.

\begin{remark}
\label{rem:grp_lang}
Note some ambiguity here in the sense that if $\Sigma = \{a,b\}$, then $(aa)^{\ast}$
is not a group language over this alphabet, but it is over the unary alphabet $\{a\}$.
Hence we mean the existence of an alphabet such that the language is recognized by a permutation automaton over this alphabet. By definition, $\{\varepsilon\}$ is considered
to be a group language\footnote{It is not possible
to give such an automaton for $|\Sigma|\ge 1$,
but allowing $\Sigma=\emptyset$ the single-state automaton
will do, or similarly as $\Sigma^{\ast} = \{\varepsilon\}$ in this case.}. 
%This will unify the statements
%of some results.
Also, group languages are closed under the boolean operations
if viewed over a common alphabet, but not over different alphabets. For instance, $L = (aa)^* \cup (bbb)^*$
is not a group language.
\end{remark}

\subsection{Commutative Aperiodic and Group Languages}

%For group languages,
%as the minimal automaton is a permutation automaton, 
%and by methods 
% A commutative regular language is a group languages if and only
% if its index vector is the zero vector. This follows by constructions 
% from~\cite{DBLP:conf/cai/Hoffmann19,Hoffmann20} and as 
% the minimal automaton of a group language is a permutation automaton.
% The period vector was only defined for commutative regular languages,
% but let us note that in our context it also makes sense
% to define it for group languages in general. In this case, the numbers $p_j \ge 1$
% are the smallest numbers such that $[a_j]_{\equiv_L}[a_j^{p_j}]_{\equiv_L} = [a_j]_{\equiv_L}$.
% As each letter acts as a permutation, this gives $[a_j^{p_j}]_{\equiv_L} = [\varepsilon]_{\equiv_L}$.
% As for each $w \in \Sigma^*$ we have $[w]_{\equiv_L}[a_j]_{\equiv_L} = [wa_j]_{\equiv_L}$
% and this gives a permutation on the equivalence classes,
% we find $[a_j]_{\equiv_L}^{p_j-1} = [\varepsilon]_{\equiv_L}$.
% Hence $p_j - 1$ is the order of the letter $a_j$ acting
% on the states of the minimal automaton.
% Note that for commutative group languages, the periods equal the order of the corresponding
% letter viewed as a permutation on the states of the minimal automaton.

The next definitions and results are taken from~\cite{DBLP:reference/hfl/Pin97,Pin86}.
For $a \in \Sigma$ and $n,r \ge 0$ set
\[
F(a,r,n) = \{ u \in \Sigma^* \mid |u|_a \equiv r \pmod{n} \},
\]
and, for $a \in \Sigma$
and $t \ge 0$,
\[
 F(a,t) = \{ u \in \Sigma^+ \mid |u|_a \ge t \}.
\]
Note that these sets are defined relative to an alphabet $\Sigma$.

\begin{example} 
\label{ex:F_sets}
Let $\Sigma$ be a non-empty alphabet, $a \in \Sigma$
and $\Gamma \subseteq \Sigma$.

\begin{enumerate}
\item $F(a, 0, 1) = \Sigma^*$.
\item $F(a, 0, 2) \cap F(a, 3, 4) = \emptyset$.
\item $F(a, 1) = \Sigma^* a \Sigma^*$.
\item $\Gamma^* = \Sigma^* \setminus \left( \bigcup_{b \in \Sigma \setminus \Gamma} F(b, 1) \right)$.
\end{enumerate}
 
\end{example}

\begin{theorem}[\cite{DBLP:reference/hfl/Pin97,Pin86}]
\label{thm:com_varities}
 Let $\Sigma$ be an non-empty\footnote{For $\Sigma = \emptyset$, we set all these classes 
 to equal $\{ \emptyset, \{\varepsilon\} \}$.} 
 alphabet.
 
 \begin{enumerate}
 \item  The class of commutative group languages over $\Sigma$ is the boolean
 algebra generated by the languages of the form $F(a, r, n)$, where $a \in \Sigma$
 and $0 \le r < n$.
 
 \item The class of commutative aperiodic languages over $\Sigma$
  is the boolean algebra generated by the languages of the form $F(a, t)$,
  where $a \in \Sigma$ and $t \ge 0$.
 
 \item The class of all commutative regular languages over $\Sigma$
  is the boolean algebra generated by the languages
  of the form $F(a, t)$ or $F(a, r, n)$, where  $t \ge 0$, $0 \le r < n$
  and $a \in \Sigma$.
 \end{enumerate}
\end{theorem}

A \emph{positive boolean algebra} is a class of sets closed
under union and intersection.
In~\cite{DBLP:reference/hfl/Pin97}, the positive variety $\mathbf{Com}^+$ was introduced. 
A \emph{positive variety}~\cite{DBLP:reference/hfl/Pin97,Pin86} $\mathcal V$
of languages maps any alphabet $\Sigma$ to a subclass $\mathcal V(\Sigma^*)$ of languages
over this alphabet that is closed under union, intersection, quotients
and inverse homomorphisms.
I only mention in passing that there is a rich theory between positive varieties of languages
and so called pseudovarieties of finite ordered semigroups~\cite{Pin86}.
Originally, $\textbf{Com}^+$ was defined in terms of certain ordered semigroups, but here, as
we do not introduce these notions, we introduce it with an equivalent characterization from~\cite{DBLP:reference/hfl/Pin97}.

\begin{definition}[\cite{DBLP:reference/hfl/Pin97}]
\label{def:com_plus}
 For every alphabet $\Sigma$, the class $\mathbf{Com}^+(\Sigma^*)$
 is the positive boolean algebra generated by the languages of the form $F(a, t)$
 and $F(a, r, n)$, where $a \in \Sigma$ and $t \ge 0$, $0 \le r < n$.
\end{definition}

%% todo shuffle properties beliebige indexmengen!

\begin{lemmarep} 
\label{lem:no_gamma}
 Let $\Sigma$ be a non-empty set\footnote{For $\Sigma = \emptyset$, we set $\mathbf{Com}^+(\Sigma^*) = 
 \{ \emptyset, \{\varepsilon\} \}$.}
 and $\Gamma \subseteq \Sigma$ be a proper subset. 
 Then, $\{ \Gamma^*, \Gamma^+ \} \cap \mathbf{Com}^+(\Sigma^*) = \emptyset$.
\end{lemmarep}
\begin{proof}
 Every language in $\mathbf{Com}^+(\Sigma^*)$
 could be written as a union over intersections of languages of the form
 $F(a, t)$
 and $F(a, r, n)$, where $a \in \Sigma$ and $t \ge 0$, $0 \le r < n$.
 Let $L \subseteq \Sigma^*$ be such an intersection
 of these languages. If $L \ne \emptyset$,
 by Theorem~\ref{thm:CRT},
 we can suppose for each $a \in \Sigma^*$
 at most one set of the form $F(a, r, n)$ for $0 \le r < n$ appears in an expression
 for $L$ as an intersection. Also, as
 $F(a, t_1) \cap F(a, t_2) = F(a, \max\{k_1, k_2\})$ for $t_1, t_2 \ge 0$
 we can also suppose for each letter at most one set 
 of the form $F(a, t)$ for $t \ge 0$
 appears in an expression for $L$ for any $a \in \Sigma$.

 Fix $a \in \Sigma$.
 As, for $b \in \Sigma$, $F(b, t) = (F(b, t) \cap b^*) \shuffle (\Sigma\setminus\{b\})^*$
 and $F(b, r, n) = (F(b, r, n) \cap b^*) \shuffle (\Sigma\setminus\{b\})^*$,
 we can then deduce that $a^* \cap L$ is non-empty, actually infinite.
 Hence, every union of such languages has this property, which gives the claim.
 In particular, no non-empty language $L \subseteq \Gamma^*$
 is in $\mathbf{Com}^+(\Sigma^*)$.~\qed
\end{proof}

Note that the previous lemma, by choosing $\Gamma = \emptyset$,
implies for $\Sigma \ne \emptyset$ that $\{\varepsilon\} \notin \mathbf{Com}^+(\Sigma^*)$.
The sets $F(a,t)$ were defined as subsets of $\Sigma^+$~\cite{DBLP:reference/hfl/Pin97}, not~$\Sigma^*$. However, this makes no difference as $\Sigma^+ = F(a, 0) = \bigcup_{b \in \Sigma} F(a,1)$ and $F(a, 0, 1) = \Sigma^*$
and so $\{ \Sigma^+, \Sigma^* \} \subseteq \mathbf{Com}^+(\Sigma^*)$.

%\section{State Complexity Results}

%\input{sc_part}

%\section{Closure Properties of Commutative Languages under Projection}

%Here, we will prove various closure properties related to commutative languages.
%The main result is the identification of a subclass closed under iterative shuffle
%and its corollaries.

\section{Commutative Aperiodic and Group Languages under Projection}
\label{subsec:proj_star_free_pure_group}

First, we strengthen Theorem~\ref{thm:com_varities} for commutative group languages.

\begin{theoremrep}
\label{thm:group_union_Fak}
 A commutative language $L\subseteq \Sigma^*$
 is a group language if and only if
 it could be written as a finite union
 of languages of the form
 \[
  \bigcap_{i=1}^m F(a_i, k_i, n_i),
 \]
 where $a_i \in \Sigma$ and $0 \le k_i < n_i$ for $i \in \{1,\ldots,m\}$ with $m \ge 0$.
\end{theoremrep}
\begin{proof}
 Let $L \subseteq \Sigma^*$ be a commutative group language.
 Then, by Theorem~\ref{thm:com_varities},
 $L$ is in the boolean algebra generated by languages of the form $F(a,n,k)$.
 First, by using DeMorgan's laws, $L$ is in the
 positive boolean closure of languages
 of the form
 $
  F(a,k,n)$ or $\overline{F(a,k,n)}.
 $
 Now,
 \[
  \overline{F(a,k,n)} 
   = \bigcup_{i \in \{0,\ldots,n-1\}\setminus \{k\}} F(a,i,n).
 \] % where a,k,n jeweils hinschreiben. todo
 Hence, we can suppose $L$
 is in the positive boolean closure of languages
 of the form $F(a,k,n)$.
 As intersection distributes over union, we can then write $L$
 as a union of intersection of languages of the form $F(a,k,n)$, i.e,
 $L$ is a union of languages of the form
 \[
  \bigcap_{i=1}^m F(a_i, k_i, n_i).
 \]
 Hence, we have shown the claim
% Verschiedene a_i!!!! aber mit chinese remainder kann man annehmen alle verschieden.
%  By the generalized Chinese Remainder Theorem, Theorem~\ref{todo}, 
%  such a set is either empty, or has itself
%  the form $F(a, k, n)$.~\qed
%
%
%  By Theorem~\ref{todo}, $L$ is in the boolean algebra
%  of these languages.
%  By using the laws of a boolean algebra, we can write
%  \[
%   L = \bigcup_{i=1}^n \bigcap_{j=1}^{m_i} K(a_{ij}, k_{ij}, n_{ij})
%  \]
%  for $a_{ij} \in \Sigma$ and non-negative numbers $k_{ij}, n_{ij}, m_i \ge 0$,
%  where $K(a_{ij}, k_{ij}, n_{ij})$
%  equals $F()$ or $\overline{F()}$.

Conversely, if $L$ is written as a union over languages of the form as
written, then, by Theorem~\ref{thm:com_varities},
it is a commutative group language.~\qed
\end{proof}

A similar statement holds for the star-free languages.
But we cannot use the languages $F(a, t)$ introduced earlier.
Set, for $a \in \Sigma$ and $k_1, k_2 \in \mathbb N_0 \cup \{\infty\}$,
 \[
  I(a, k_1, k_2) = \{ u \in \Sigma^* \mid k_1 \le |u|_a < k_2 \}. 
 \]

% todo oben in formel aufpassen, n vs n_i

\begin{theoremrep} %todo einheitlich star-free oder aperiodic schreiben.
\label{thm:aperiodic_union_Fakk}
 A commutative language $L \subseteq \Sigma^*$ is aperiodic if and only if
 it could be written as a finite union
 of sets of the form 
 \[ 
  \bigcap_{i=1}^n I(a_i, r_i, s_i),
 \]
 where $0 \le r_i < s_i$ and $a_i \in \Sigma$ for $i \in \{1,\ldots,n\}$ with $n \ge 0$.
 % with $|\{a_1, \ldots, a_n\}| = n$. todo bonus, remark alle verschieden
\end{theoremrep}
\begin{proof}
 We use the characterization stated in Theorem~\ref{thm:com_varities}.
 First, let $L \subseteq \Sigma^*$ be a commutative and star-free language.
 We have 
 \begin{align*} 
  \overline{F(a,k)} & = \{\varepsilon\} \cup \{ u \in \Sigma^* \mid |u|_a < k \} = I(a, 0, k), \\
  F(a,k)            & = \left\{ 
  \begin{array}{ll}
   I(a, k, \infty) & \mbox{if } k > 0; \\
   \bigcup_{b\in \Sigma} I(b, 1, \infty) & \mbox{if } k = 0. % = \Sigma^*\setminus\{eps\}
  \end{array}\right.
 \end{align*}
 Hence, $L$ is in the positive boolean algebra generated
 by sets of the form $I(a, r, s)$.
 As intersection distributes over union, we can then write $L$
 as a union of intersections of languages of the form $I(a, r, s)$, i.e.,
 $L$ is a union of languages of the form
 \[
  \bigcap_{i=1}^n I(a_i, r_i, s_i).
 \]
 Conversely, suppose $L$ has the form as written in the statement.
 Then,
 \begin{align*}
     I(a, r, s) = \left\{
     \begin{array}{ll}
       F(a, r) \cap \overline{F(a,s)} & \mbox{if } r > 0, s \ne \infty; \\
       F(a, r)                        & \mbox{if } r > 0, s = \infty; \\
       \overline{F(a, s)}             & \mbox{if } r = 0, s \ne \infty; \\
       \Sigma^*                       & \mbox{if } r = 0, s = \infty.
     \end{array}\right.
 \end{align*}
 As $\Sigma^* = \overline{(F(a, 0) \cap \overline{F(a, 0)})}$,
 we find that $L$ is in the boolean closure
 of languages of the form $F(a,k)$.
 Hence, by Theorem~\ref{thm:com_varities}, $L$ is a commutative star-free language.~\qed
\end{proof}

Next, we state how these languages
behave under projection.

\begin{lemmarep} 
\label{lem:I_sets_proj}
Let $\Gamma \subseteq \Sigma$, $n \ge 0$, $a_i \in \Sigma$
and $0 \le r_i < s_i$ for $i \in \{1,\ldots,n\}$.
Then,
 \[
  \pi_{\Gamma}\left( \bigcap_{i=1}^n I(a_i, r_i, s_i) \right) 
   = \left( \bigcap_{\substack{i \in \{1,\ldots,n\} \\ a_i \in \Gamma}} I(a_i, r_i, s_i) \right) \cap \Gamma^*.
 \]
\end{lemmarep}
\begin{proof}
 Let $a \in \Sigma$ and $0 \le k_1 < k_2$.
 Then, for $a \in \Gamma$ we have
 $k_1 \le |\pi_{\Gamma}(u)|_a < k_2$
 if and only if $k_1 \le |u|_a < k_2$.
 The letters not in $\Gamma$ are deleted
 and do not appear in the image words,
 hence do not contribute to the result.~\qed
\end{proof}

With Lemma~\ref{lem:I_sets_proj}, we can prove that the star-free commutative languages
are closed under projections.

\begin{propositionrep}
\label{prop:proj_star_free}
 Let $L \subseteq \Sigma^*$ be commutative and star-free.
 Then, for any $\Gamma \subseteq \Sigma$,
 the language $\pi_{\Gamma}(L)$ is commutative star-free.
\end{propositionrep}
\begin{proof}
 By Theorem~\ref{thm::comm_lang_closure}, 
 the projected language is commutative.
 By Theorem~\ref{thm:com_varities} and as, for $U, V \subseteq \Sigma^*$, 
 \[
  \pi_{\Gamma}(U \cup V) = \pi_{\Gamma}(V) \cup \pi_{\Gamma}(V)
 \]
 and the star-free languages are closed under union and intersection,
 we only need to show, by Theorem~\ref{thm:aperiodic_union_Fakk}
 and Lemma~\ref{lem:I_sets_proj},
 that $\Gamma^*$ is star-free.
 But this is shown in Example~\ref{ex:F_sets}.~\qed
\end{proof}

In general, for homomorphic mappings, this is not true, as $a^*$ could be mapped
homomorphically onto $(aa)^*$, and $(aa)^*$ is not star-free~\cite{McNaughton71}.
Also, more specifically, there exist non-commutative
star-free languages with a non-star-free projection language. For example,
the language $L = (aba)^*$ is star-free, as
\[
 L = \{\varepsilon\} \cup (aba\Sigma^* \cap \Sigma^*aba) \setminus (\Sigma^*\cdot\{aaa,bba,bab,abb\}\cdot\Sigma^*),
\]
but $\pi_{\{a\}}(L) = (aa)^*$.
Similarly, with Theorem~\ref{thm:group_union_Fak},
we can show the next result.

\begin{propositionrep}
 Let $L \subseteq \Sigma^*$ be a commutative group language. 
 Then, for any $\Gamma \subseteq \Sigma$,
 the language $\pi_{\Gamma}(L)$ is a commutative group language.
\end{propositionrep}
\begin{proof}
 The proof is similar to the proof of Proposition~\ref{prop:proj_star_free},
 but using
 Theorem~\ref{thm:group_union_Fak},
 a similar property for the intersection
 of sets of the form $F(a, r, n)$
 and the fact that $\Gamma^*$
 is a group language when considering
 $\Gamma$ as the whole alphabet in the image of $\pi_{\Gamma}$ (but not 
 when it is a proper subalphabet of $\Sigma$,
 compare Remark~\ref{rem:grp_lang}).~\qed
\end{proof}

However, also here, this is false for general group languages.
The language $(aa)^*$ could be mapped homomorphically
onto $L = (abab)^*$, which is not a group language. 
Also, for projections, consider the group language given by the permutation
automaton
$\mathcal A = (\{a,b\}, \{0,1,2\}, \delta, 0, \{2\})$
with $\delta(0,a) = 1$, $\delta(1,a) = 0$, $\delta(2,a) = 2$
and $\delta(0,b) = 1$, $\delta(1,b) = 2$, $\delta(2,b) = 0$.
Then, $\pi_{\{b\}}(L(\mathcal A)) = bb^*$, which is not a group language. 
For example, $b$ is the projection of $ab \in L(\mathcal A)$,
or $bbb$ the projection of $abbab \in L(\mathcal A)$.

\section{A Class of Regular Languages Closed under Iterated Shuffle}
\label{subsec:it_shuffle}

% beispiele, wann nicht a shuffle {b,bb} oder a shuffle b

Here, we introduce a subclass of commutative regular languages,
which contains the commutative group languages,
that is closed under iterated shuffle.
In Definition~\ref{def:diagonal_periodic_lang},
we introduce the \emph{diagonal periodic} languages, and first
establish that the iterated shuffle of such a language
gives a language that is a union of diagonal periodic languages.
We then use this result to show closure under this operation
of our subclass, which either could be described as
the positive boolean algebra generated by languages of the form $F(a, n, k)$, $F(a, k)$, $\Gamma^*$
and $\Gamma^+$ for $\Gamma\subseteq \Sigma$, $a \in \Sigma$, $0 \le k < n$,
or as finite unions of diagonal periodic languages.

Note that, for already very simple languages, the iterated shuffle can give
non-regular languages, for example
$
  (a \shuffle b)^{\itshuffle}         = \{ab,ba\}^{\itshuffle}  = \{ u \in \{a,b\}^* \mid |u|_a = |u|_b \}$,
  or
$(a \shuffle \{b,bb\})^{\itshuffle}  = \{ u \in \{a,b\}^* \mid |u|_b \le |u|_a \le 2|u|_b \}$.

% name diagonal periodic begründen

\begin{definition} % remark name, parikh ubilder von rechtecken, base characterisation
\label{def:diagonal_periodic_lang}
 A \emph{diagonal periodic} language over $\Gamma \subseteq \Sigma$ is a language
 of the form % vector von base in keiner achsen-parallelen hyperebene
 \[
  \bigshuffle_{a \in \Gamma} a^{k_a} (a^{p_a})^*,
 \]
 where $k_a \ge 0$ and $p_a > 0$ for $a \in \Gamma$ when $\Gamma \ne \emptyset$,
 or the language $\{\varepsilon\}$.
\end{definition}

% By setting $\Gamma = \emptyset$, we find that $\{\varepsilon\}$
% is a diagonal periodic language.

\begin{remark}
%  The periodic language were introduced in~\cite{EhrenfeuchtHR83},
%  and characterized with automata models in~\cite{hoffmann2x}.
 Let $\Sigma = \{a_1, \ldots, a_k\}$
 In~\cite{EhrenfeuchtHR83} a sequence of vectors $\rho = v_0, v_1, \ldots, v_k$ from $\mathbb N_0^k$
was called a \emph{base} if $v_i(j) = 0$ for\footnote{Note that the entries of
$v \in \mathbb N_0^k$ are numbered by $1$ to $k$, i.e., $v = (v(1),\ldots, v(k))$.}
$i,j \in \{1,\ldots,k\}$ such that $i \ne j$.  The $\rho$-set was defined
as
$
 \Theta(\rho) = \{ v \in \mathbb N^k : v = v_0 + l_1 v_1 + \ldots + l_k v_k \mbox{ for some } l_1, \ldots, l_k \in \mathbb N_0 \}.
$
 Then, in~\cite{EhrenfeuchtHR83}, a language $L \subseteq \Sigma^*$
 was called \emph{periodic} if, for some fixed order $\Sigma = \{a_1, \ldots, a_k\}$,
 there exists a base $\rho$ such that
 $
  L = \psi^{-1}(\Theta(\rho)).
 $
 With this geometric view, the diagonal periodic languages
 are those periodic languages such that, for $i,j \in \{1,\ldots,k\}$, 
 either
 \[
  v_i(j) \ne 0 \mbox{ or } v_i(j) = v_0(j) = 0.
 \]
 Intuitively, and very roughly, the vector $\sum_{a_i \in \Gamma} v_i$ points diagonally in the
 subspace corresponding to the letters in $\Gamma$, or more precisely,
 the dimension of the subspace spanned by $v_1, \ldots, v_k$ is precisely $|\Gamma|$.
 Hence, the name diagonal periodic.
\end{remark}

As the languages $a^{k_a}(a^{p_a})^*$, $a \in \Gamma$, are regular
and the binary shuffle operation is regularity-preserving~\cite{Ito04},
we get the next result. But it was also established in~\cite{EhrenfeuchtHR83,DBLP:conf/cai/Hoffmann19,Hoffmann20}
for the more general class of periodic languages.

\begin{proposition}
\label{prop:diag_periodic_regular}
 The diagonal periodic languages are regular and commutative.
\end{proposition}

\begin{remark} 
 Suppose, for each $a \in \Sigma$, we have a unary language
 $L_a \subseteq a^*$ and $\Gamma \subseteq \Sigma$.
 Then, $\pi_{\Gamma}(\bigshuffle_{a \in \Sigma} L_a) = \bigshuffle_{a \in \Gamma} L_a$
 % = ( \bigshuffle_{a \in \Sigma} L_a ) \cap \Gamma^*$
 % ne, gilt nicht, a\shuffle b \cap b^* = emptyset, ungleich projektion
 and $\pi_{\Sigma}^{-1}( \bigshuffle_{a \in \Gamma} L_a ) =  \bigshuffle_{a \in \Gamma} L_a \shuffle (\Sigma\setminus\Gamma)^*$.
 This could be worked out to give a different
 proofs for the results from Subsection~\ref{subsec:proj_star_free_pure_group}.
\end{remark} 

\begin{remark}
 The reason a subalphabet $\Gamma \subseteq \Sigma$ is included in Definition~\ref{def:diagonal_periodic_lang}, 
 and later in the statements, 
 is due to Lemma~\ref{lem:no_gamma}, i.e., to have a larger
 class as given by $\mathbf{Com}^+$.
\end{remark}

 Next, we investigate what languages we get if we apply
 the iterated shuffle to diagonal periodic languages.

\begin{propositionrep} % oder lemma? und positive boolean algebra wenn vereinigung von diesen.
\label{prop:diagonal_periodic_it_shuffle_reg}
 The iterated shuffle of a diagonal periodic
 language $L \subseteq \Sigma^*$ over $\Gamma\subseteq \Sigma^*$ is
 a finite union of diagonal periodic languages. In particular,
 it is regular.
\end{propositionrep} % remark, genau ein p_i = 0 erlaubt, und es wäre auch noch abgeschlossen
\begin{proof}
 Let $L \subseteq \Sigma^*$ be a diagonal periodic language over $\Gamma \subseteq \Sigma$.
 Write
  \[
  L = \bigshuffle_{a \in \Gamma} a^{k_a} (a^{p_a})^*
 \]
 for numbers $k_a \ge 0$, $p_a > 0$ with $a \in \Gamma$.
 Now, by Theorem~\ref{thm:shuffle_properties} and as
 for unary language concatenation and shuffle coincide,
 we find $\bigshuffle_{i=1}^m L = \bigshuffle_{a \in \Gamma} a^{m\cdot k_a} (a^{p_a})^*$.
 So,
 \[
  L^{\itshuffle} = \{\varepsilon\}\cup \bigcup_{m > 0} \bigshuffle_{a \in \Gamma} a^{m\cdot k_a} (a^{p_a})^* .
 \]
 %Let $a \in \Gamma$ and let $m > 0$ and $r_a \ge 0$
 %and consider the sum $m \cdot k_a + r_a p_a$
 Hence, $u \in L^{\itshuffle}$ if and only if there exists $r_a \ge 0$, $a \in \Gamma$,
 such that $u \in \bigshuffle_{a \in \Gamma} a^{m \cdot k_a + r_a \cdot p_a}$ for some $m > 0$.
 Now, fix $a \in \Gamma$ and consider
 the sum $m \cdot k_a + r_a \cdot p_a$.
 We have, for any $t \in \mathbb Z$, 
 \[
  m \cdot k_a + r_a \cdot p_a = (m - tp_a) \cdot k_a + (r_a + tk_a) \cdot p_a.
 \]
 In particular, we can choose $t \in \mathbb Z$
 such that $1 \le m - tp_a \le p_a$ and $r_a + tk_a \ge 0$.
 Hence,
 \[
  \bigcup_{m > 0} a^{m\cdot k_a} (a^{p_a})^* = \bigcup_{i = 1}^{p_a} a^{i \cdot k_a} (a^{p_a})^*.
 \]
 Let $N$ be the least common multiple of the numbers $p_a$, $a \in \Gamma$.
 Suppose 
 \[ u \in \bigcup_{m > 0} \bigshuffle_{a \in \Gamma} a^{m\cdot k_a} (a^{p_a})^*.
 \]
 Then, there exist numbers $r_a \ge 0$, $a \in \Gamma$, and $m > 0$
 such that
 \[
  u = \bigshuffle_{a \in\Gamma} a^{m \cdot k_a + r_a \cdot p_a}.
 \]
 Similarly as above, for any $a \in \Gamma$ and $t \ge 0$,
 we have
 \[
  m \cdot k_a + r_a \cdot p_a = (m - tN) \cdot k_a + \left(r_a + t \frac{N}{p_a} k_a \right) \cdot p_a.
 \] 
 So, we can choose $t \ge 0$
 such that $1 \le m - tN \le N$ and
 \[
  u = \bigshuffle_{a \in\Gamma} a^{(m - tN) \cdot k_a + \left(r_a + t \frac{N}{p_a} k_a \right) \cdot p_a}
   \in \bigcup_{i=1}^N  \bigshuffle_{a\in \Gamma} a^{i\cdot k_a}(a^{p_a})^*.
 \] 
 Hence, we have shown
 $
   L^{\itshuffle} \subseteq \{\varepsilon\} \cup \bigcup_{i=1}^N \bigshuffle_{a\in \Gamma} a^{i\cdot k_a}(a^{p_a})^*.
 $
 The other inclusion is obvious, and we find 
 \[
  L^{\itshuffle} = \{\varepsilon\} \cup \bigcup_{i=1}^N  \bigshuffle_{a\in \Gamma} a^{i\cdot k_a}(a^{p_a})^*.
 \]
 So, as a finite union of diagonal periodic languages, hence regular languages by Proposition~\ref{prop:diag_periodic_regular},
 the language $L^{\itshuffle}$ is itself regular.~\qed
\end{proof}

% mk_i + r_i p_i = (m - p_i) k_i + (r_i + k_i) p_i 
%
% Kann stets 0 \le m < p_i annehmen, d.h. vereinigung ab da abbrechnen.
% also nach max{p_i}

% bedingung kein a kann man gar nicht ausrücken... aber in obigen shuffle produkten ist es möglich. -> remark
% kein \Gamma^* -> noch hinzufügen! todo

The next lemma is the link between the languages $F(a, t)$, $t \ge 0$,
and $F(a, r, n)$, $0 \le r < n$, and
the diagonal periodic languages.

\begin{lemmarep} 
\label{lem:diagonal_periodic_other_form}
% wenn man Sigma^* in den F, dann könnte man einen schnitt über
% a's machen. neural F(a, 0) und F(a,0,1)
% lemma, neue mengen K(a,k) und zeigen dass gleich. oder nur remark
 Let $\Sigma_1, \Sigma_2 \subseteq \Sigma$.
 Suppose we have numbers $t_a$ for $a \in \Sigma_1$
 and $0 \le r_a < n_a$ for $a \in \Sigma_2$.
 Then,
 \[
  \bigcap_{a \in \Sigma_1} F(a, t_a) \cap \bigcap_{a \in \Sigma_2} F(a, r_a, n_a)
   = \bigshuffle_{a \in \Sigma} a^{k_a}(a^{p_a})^*,
 %  = \left( \bigshuffle_{a \in \Sigma_1 \cup \Sigma_2} a^{k_a}(a^{p_a})^* \right) \shuffle (\Sigma \setminus (\Sigma_1 \cup \Sigma_2))^*
   %\left( \bigshuffle_{a \in \Sigma\setminus (\Sigma_1 \cup \Sigma_2)} a^* \right)
    % das als remark, letzte beide gleichheiten.
 \] % man kann annehmen Sigma = Sigma_1 \cup Sigma_2, sonst fehlende reinshufflen mit jeweils a^* \shuffle b^* usw...
 where\footnote{For $x,n \in \mathbb N$, by $x \bmod n$ we denote the unique number $r \in \{0,\ldots,n-1\}$
 such that $r \equiv x \pmod{n}$.}
 \[
  k_a = \left\{ 
  \begin{array}{ll}
   t_a + (n_a - ((t_a - r_a) \bmod n_a)) & \mbox{if } a \in \Sigma_1 \cap \Sigma_2, t_a > r_a; \\
   r_a                           & \mbox{if } a \in \Sigma_1 \cap \Sigma_2, t_a  \le   r_a; \\
   r_a                           & \mbox{if } a \in \Sigma_2 \setminus \Sigma_1; \\
   t_a                           & \mbox{if } a \in \Sigma_1 \setminus \Sigma_2; \\
   0                             & \mbox{if } a \notin \Sigma_1 \cup \Sigma_2.
  \end{array}
  \right.
\]
and
$
 %\quad\mbox{ and }\quad
  p_a = \left\{ 
   \begin{array}{ll}
   n_a & \mbox{if } a \in \Sigma_2; \\
   1   & \mbox{if } a \notin \Sigma_2. 
   \end{array}
  \right.
$
\end{lemmarep}
\begin{proof}
 For the first case, let us first state an auxiliary claim. Let $t,r,n \ge 0$
 with $t > r$ and $0 \le r < n$.
 
 \begin{claiminproof}
  For the unique number $m \ge 0$
  with $r + m\cdot n < t \le r + (m+1)\cdot n$
  we have $(m+1)n = t + (n - ((t-r) \bmod n))$.
 \end{claiminproof}
 \begin{claimproof}
  As $mn < t - r \le (m+1)n$,
  we have $mn + ((t - r) \bmod n) = t - r$.
  Hence,
  \begin{align*} 
   (m+1)n & = mn + ((t-r) \bmod n)) + (n - ((t-r) \bmod n)) \\
          & = t - r + (n - ((t - r) \bmod n)),
  \end{align*}
  which gives the claim.
 \end{claimproof}
 We have, for any $a \in \Sigma$,
 \begin{align*}
     F(a, t) \cap F(a, r, n) & = (a^{t}a^* \cap a^r(a^n)^*) \shuffle (\Sigma \setminus \{a\})^*. 
 \end{align*}
 And, by the above claim,
 \[
 a^{t}a^* \cap a^r(a^n)^*
  = \left\{ 
   \begin{array}{ll}
    a^r(a^n)^* & \mbox{if } t \le r; \\
    a^{t + (n - ((t - r) \bmod n)} (a^n)^* & \mbox{if } t > r.
   \end{array}
  \right.
 \]
 Furthermore,
 \begin{align*}
     F(a, r, n) & = a^r(a^n)^* \shuffle (\Sigma \setminus \{a\})^*; \\
     F(a, t)    & = a^ta^*.
 \end{align*}
 And these formulas give the claim.~\qed

% auf die form mit schnitt threshold + generalized chinese remaineder    
\end{proof}

Now, we have everything together to prove our main theorem
of this subsection.

\begin{theorem} 
\label{thm:it_shuffle_reg_preserving}
 Let $L \subseteq \Sigma^*$ be in the positive boolean algebra
 generated by languages of the form 
 $F(a, k)$, $F(a, k, n)$, $\Gamma^+$ and  $\Gamma^*$ for $\Gamma \subseteq \Sigma$.
 %and  $\{\varepsilon\} \cup F(a,k), k = 0$. % F(a,0,1) = Sigma^*
 Then, the iterated shuffle of $L$
 is contained in this positive boolean algebra. In particular, 
 the iterated shuffle is regular.
\end{theorem} % einheitlich t und r statt k nutzen
\begin{proofsketch}
 As intersection distributes over union,
 $L$ could be written as an intersection over the generating languages.
 Now, \[ 
 F(a, k_1) \cap F(a, k_2) = F(a, \max\{k_1, k_2\})
 \]
 and, by the generalized Chinese Remainder Theorem, Theorem~\ref{thm:CRT},
 every intersection $\bigcap_{i=1}^m F(a, r_i, n_i)$
 is either the empty set, or also a set of the form $F(a, r, n)$.
 So, every such intersection could be written
 in the form
 \[
  \left( \bigcap_{a \in \Sigma_1} F(a, t_a) \right) \cap 
  \left( \bigcap_{a \in \Sigma_2} F(a, r_a, n_a) \right) \cap L
 \] 
 where $L \in \{\Gamma^+, \Gamma^*\}$
 for some $\Gamma \subseteq \Sigma$ and $\Sigma_1, \Sigma_2 \subseteq \Sigma$.
 By Lemma~\ref{lem:diagonal_periodic_other_form}, these language
 are diagonal periodic over $\Gamma$.
 By Theorem~\ref{thm:shuffle_properties}, the iterated shuffle
 of $L$ is a finite shuffle product of iterated shuffles of these languages,
 which are regular by Proposition~\ref{prop:diagonal_periodic_it_shuffle_reg}.
 Hence, they are a finite shuffle product of regular languages and
 as the binary shuffle product is a regularity-preserving operation~\cite{Ito04}, 
 the language $L$ is regular. More precisely, as the iterated shuffles are finite unions
 of diagonal periodic languages, the result could
 be written as a finite union of diagonal periodic languages, which, by Lemma~\ref{lem:diagonal_periodic_other_form},
 are contained in this class.~\qed
\end{proofsketch}

The method of proof of Theorem~\ref{thm:it_shuffle_reg_preserving} also gives
the next result.

\begin{proposition}
\label{prop:form_class}
 The positive boolean algebra generated by languages
 of the form $F(a,k)$, $F(a,k,n)$, $0 \le k < n$, $\Gamma^+$
 and $\Gamma^*$, $\Gamma \subseteq \Sigma$,
 is precisely the language class of finite unions
 of the diagonal periodic languages.
\end{proposition}

\begin{corollary}
 The iterated shuffle of a commutative group language
 is regular. % nach automat fragen. todo open problem
\end{corollary}
\begin{proof} By Theorem~\ref{thm:group_union_Fak},
the class introduced in Theorem~\ref{thm:it_shuffle_reg_preserving}
contains the group languages.~\qed \end{proof}

\begin{corollaryrep} 
 The variety $\textbf{Com}^+$
 is closed under iterated shuffle.
\end{corollaryrep}
\begin{proof}
By Definition~\ref{def:com_plus},
the class introduced in Theorem~\ref{thm:it_shuffle_reg_preserving}
contains $\mathbf{Com}^+(\Sigma^*)$ for any alphabet 
$\Sigma$. Furthermore, by the method of proof of Theorem~\ref{thm:it_shuffle_reg_preserving} and as the iterated
shuffle does not introduce new letters, and does not remove old letters, we do not leave
the class $\mathbf{Com}^+(\Sigma^*)$.~\qed 
\end{proof} 

Also, as, for $U_a, V_a \subseteq \{a\}^*$, $a \in \Sigma$, we have $(\bigshuffle_{a \in \Sigma} U_a) \shuffle (\bigshuffle_{a \in \Sigma} V_a)
= (\bigshuffle_{a \in \Sigma} (U_a \cdot V_a))$,
and with Theorem~\ref{thm:shuffle_properties},
we can deduce, by Proposition~\ref{prop:form_class},
the next result. This extends an old result by J.F. Perrot~\cite{DBLP:journals/tcs/Perrot78}
stating that the star-free commutative language are closed under binary shuffle.

\begin{proposition}
 The positive boolean algebra generated by the languages
 $F(a,k)$, $F(a,k,n)$, $0 \le k < n$, $\Gamma^+$
 and $\Gamma^*$ for $\Gamma \subseteq \Sigma$
 is closed under binary shuffle.
\end{proposition}

\section{Characterizing Regularity of the Iterated Shuffle}% for Aperiodic Commutative Languages}

First, in Subsection~\ref{subsec:fin_lang}, we will give a necessary and sufficient
condition when the iterated shuffle of a commutative finite language is regular.
Then, in Subsection~\ref{subsec:aperiodic_group}, we will present partial results for aperiodic commutative language.
Lastly, in Subsection~\ref{subsec:dec_procedures}, we discuss
decision procedures related to regularity, the commutative closure and the iterated shuffle.

\subsection{Finite Commutative Languages}
\label{subsec:fin_lang}

Here, we investigate finite commutative languages.

% wegen perm(L)^itshuffle = perm(L^*)
% beliebige finite, ist perm(L)^{it shuffle} regulär?
% damit entscheidbar ob perm(L^*) für L endlich regulär [d.h. L regex ohne stern]
% also star height at most one perm regulär entscheidbar und unino-free!

\begin{theoremrep}
\label{thm:finite_lang_it_shuffle}
% entscheidungsverfahren in P
%  \[
%   L = \bigcup_{i=1}^n \perm(u_i) \shuffle \Gamma_i^*
%  \]
% das doch noch ein wenig komplizierter
%\[
% L = \perm(u_1^*) \shuffle \ldots \shuffle \perm(u_n^*) % spezialfall a^* shuffle ... \shuffle \Gamma^*
%\] % spezifall fall wenn Gamma hinten alle buchstaben enthält, weil dann basisvektoren.
%einfach das entscheiden.
Let $L \subseteq \Sigma^*$ be a finite language. Then, 
$\perm(L)^{\itshuffle}$
is regular if and only if for any $a \in \Sigma$ with $\Sigma^*a\Sigma^* \cap L \ne \emptyset$ we
have $a^+ \cap L \ne \emptyset$.
\end{theoremrep} 
\begin{proof}

 Let $L = \{u_1, \ldots, u_n\}$. Then, by Theorem~\ref{thm:shuffle_properties} and Theorem~\ref{thm:perm_semiring_hom},
 \begin{equation}\label{eqn:fin_it_shuffle} % todo explizit gleichungskette hinschreiben
  \perm(L)^{\itshuffle} = \perm(u_1^*) \shuffle \ldots \shuffle \perm(u_n^*).
 \end{equation}

 Suppose that for $a \in \Sigma$, we find $u \in L$ with $|u|_a > 0$ but $a^+ \cap L = \emptyset$.
 Let $b \in \Sigma \setminus \{a\}$ be such that $|u|_b > 0$.
 Set $P = \pi_{\{a,b\}}(L)$, $m = \max\left\{ \frac{|u|_a}{|u|_b} \mid v \in P \right\} \in \mathbb Q$
 % P kommutativ nutzen? todo zitieren
 and choose $w \in P$ with $|w|_b m = |w|_a$.
 Let $0 < s < t$, then, 
 as $\frac{tm}{s} > m$ and by the maximal choice of $m$,
 \[
  a^{s|w|_a} b^{|w|_b s} \in P^{\itshuffle} \mbox{ and } 
  a^{t|w|_a} b^{|w|_b s} \notin P^{\itshuffle}.
 \]
 So, the words $a^{s|w|_a}$ and $a^{t|w|_a}$ are not equivalent for the Nerode right-congruence.
 Hence, the infinitely many words $\{ a^{|w|_a r} \mid r > 0 \}$
 are pairwise non-equivalent for the Nerode right-congruence of $P^{\itshuffle}$
 and so $P^{\itshuffle}$ has infinitely many distinct Nerode right-congruence classes
 and so is not regular. As $P^{\itshuffle} = \pi_{\{a,b\}}(L)^{\itshuffle} = \pi_{\{a,b\}}(L^{\itshuffle})$, % die gleichungen zeigen.
 we find that $L^{\itshuffle}$ is not regular.

 Now, suppose the condition is true. %Set $V = \{ \psi(u_i) : i \in \{1,\ldots,n\}$.
 By Equation~\eqref{eqn:fin_it_shuffle}, we have
 \begin{equation}\label{eqn:fin_lang_lin_comb}
  v \in \psi(\perm(L)^{\itshuffle} )
  \Leftrightarrow 
  \exists c_1, \ldots, c_n \in \mathbb N_0 : v = c_1 \psi(u_1) + \ldots + c_n \psi(u_n).
 \end{equation}
 Next, we select $k$ unary words from $\{u_1, \ldots, u_n\}$
 such that for every letter used in $L$ we have exactly one such non-empty word over this letter
 in this set of selected words. We assume these to be the first $k$ words
 among $u_1, \ldots, u_n$. More precisely, without loss of generality, let $1 \le k \le n$ be such that 
 for any $u_i$, $i \in \{1,\ldots,k\}$, we have $u_i \in a^+$ for some $a \in \Sigma$
 and for any $a \in \Sigma$ with $\Sigma^* a \Sigma^* \cap L \ne \emptyset$
 we have $|\{ u_1, \ldots, u_k \} \cap a^+| = 1$.
 Then, for any $i \in \{1,\ldots,k\}$, 
 we can write $\psi(u_i) = m_i \cdot \psi(a)$, where $u_i \in a^+$ and $m_i > 0$.
 Also, denote by $a_i \in \Sigma$ the letter such that $\psi(u_i) = m_i \cdot \psi(a_i)$.
 Then, for $i,j \in \{1,\ldots, k\}$, by the assumptions, $u_i \ne u_j$ implies $a_i \ne a_j$
 and $L \subseteq \{a_1, \ldots, a_k\}^*$.
 If, for $i \in \{k+1, \ldots, n\}$, % in equation oben (oben)
 we have $c_i \ge m_1 \cdots m_k$ in Equation~\eqref{eqn:fin_lang_lin_comb}, then,
 if we select number  $x_a \ge 0$, $a \in \Sigma$, such that $\psi(u_i) = \sum_{a \in \Sigma} x_a \psi(a) = \sum_{j=1}^k x_{a_j} \psi(a_j)$,
 as
 \begin{align*}
    &  x_{a_1} \frac{m_1 \cdots m_k}{m_1} \psi(u_1) + \ldots +  x_{a_k} \frac{m_1 \cdots m_k}{m_k} \psi(u_k) \\
      & =  x_{a_1} \frac{m_1 \cdots m_k}{m_1} m_1 \psi(a_1) + \ldots +  x_{a_k} \frac{m_1 \cdots m_k}{m_k} m_k \psi(a_k) \\
      & =  m_1 \cdots m_k \psi(u_i) 
     \end{align*}
 we have
 \begin{multline*}
   \left( c_1 + x_{a_1} \frac{m_1 \cdots m_k}{m_1} \right) \psi(u_1) + \ldots + \left( c_k + x_{a_k} \frac{m_1 \cdots m_k}{m_k} \right) \psi(u_k) + \ldots + \\ c_{i-1} \psi(u_{i-1}) + (c_i - m_1 \cdots m_k) \psi(u_i) + c_{i+1}\psi(u_{i+1}) + \ldots + c_n \psi(u_n).
 \end{multline*}

 Hence, we can choose the coefficients in Equation~\eqref{eqn:fin_lang_lin_comb} 
 such that, for any $i \in \{k + 1, \ldots, n\}$,
 \[ 
  c_i < m_1 \cdots m_k.
 \]
 As, by Theorem~\ref{thm::comm_lang_closure}, $\perm(L)^{\itshuffle}$ is commutative, we have, for $w \in \Sigma^*$,
 \[ 
  w \in \perm(L)^{\itshuffle} \Leftrightarrow \psi(w) \in \psi(\perm(L)^{\itshuffle}
 \]
 and, for $c_1, \ldots, c_n \ge 0$,
 \[
  \psi(w) =  c_1 \psi(u_1) + \ldots + c_n \psi(u_n)
  \Leftrightarrow 
  w \in \perm(u_1^{c_1}) \shuffle \ldots \shuffle \perm(u_n^{c_n}).
 \]
 By the previous reasoning, we conclude
 \begin{multline*}
   \perm(u_1^*) \shuffle \ldots \shuffle \perm(u_n^*) = \\
   \bigcup_{\substack{(c_{k+1}, \ldots, c_n) \\ 0 \le c_i < m_1 \cdots m_k}} \perm(u_1^*) \shuffle \ldots \shuffle \perm(u_k^*)
  \shuffle \perm(u_{k+1}^{c_{k+1}}) \shuffle \ldots \shuffle \perm(u_n^{c_n}).
 \end{multline*}
 As, for $i \in \{1,\ldots,k\}$, we have $u_i \in a_i^+$, $\perm(u_i^*) = u_i^*$.
 So, as the binary shuffle operation is regularity-preserving~\cite{Ito04},
 every part of the union is regular and as the union is finite
 we find that $\perm(L)^{\itshuffle}$
 is regular.~\qed
\end{proof} 
% remark allgemeiner sogar vereinigungen perm(u_i) shuffle perm(?
% -> damit hätte man beliebige aperiodische, auch argument "großer" teil hinten.

By the next corollary, we find that we can characterize regularity
of expressions, for instance,
of the form
\begin{align*}
    \perm(u_1^{+}) \shuffle \ldots \shuffle \perm(u_n^+) & =  \perm(u_1\cdots u_n) \shuffle \perm(u_1^*) \shuffle \perm(u_n^*) \\
    & = \perm(u_1\cdots u_n) \shuffle \perm(\{u_1, \ldots, u_n\})^{\itshuffle} 
\end{align*}
with Theorem~\ref{thm:finite_lang_it_shuffle}, where the above equalities
are implied by Theorem~\ref{thm:shuffle_properties} and Theorem~\ref{thm:perm_semiring_hom}.

\begin{corollaryrep}\label{cor:finite_lang_it_shuffle}
 Let $u \in \Sigma$ and $L \subseteq \Sigma^*$ be a finite language.
 Then, $\perm(u) \shuffle \perm(L)^{\itshuffle}$
 is regular if and only if for any $a \in \Sigma$
 with $\Sigma^* a \Sigma^* \cap L \ne \emptyset$,
 we have $a^+ \cap L \ne \emptyset$.
\end{corollaryrep}
\begin{proof}
 If $U \subseteq \Sigma^*$ is any commutative language and $u \in \Sigma^*$,
 then $\perm(u) \shuffle U$ is regular if and only 
 if $U$ is regular. One implication is clear as the binary shuffle operation
 is regularity-preserving~\cite{Ito04}.
 
 For the other implication, first note that $U \subseteq u^{-1}(\perm(u) \shuffle U)$.
 Now we argue that $U \supseteq u^{-1}(\perm(u) \shuffle U)$ holds true.
 If $x \in \Sigma^*$ is such that $ux \in \perm(u)\shuffle U$,
 then there exists $y \in U$ such that $ux \in \perm(uy)$.
 This implies $x \in \perm(y)$ and so, as $U$ is commutative, $x \in U$.
 Hence, we find
 \[ % todo quotienten einführen
  U = u^{-1} (\perm(u) \shuffle U)
 \]
 and as the quotient by a word is a regularity-preserving operation,
 the other implication follows.~\qed
\end{proof}

\subsection{Aperiodic Commutative Languages}
\label{subsec:aperiodic_lang}

Here, we investigate aperiodic commutative languages.

\begin{propositionrep}\label{prop:aperiodic_lang_form}
 Every aperiodic commutative language
 could be written as a finite union of languages
 of the form
 $
  perm(u) \shuffle \Gamma^*
 $
 for $u \in \Sigma^*$ and $\Gamma \subseteq \Sigma$.
\end{propositionrep} % und effektiv
% entweder ST 3/2 höchstens, oder direkt zeigen und
% nur remark, siehe dfa-intersection paper, todo
\begin{proof} % todo set oder language, nicht durcheinander bringen, todo s_i = infty dort hinschreiben.

 For the sets from Section~\ref{subsec:proj_star_free_pure_group} we have, with $a \in \Sigma$
 and $r_1, s_1, r_2, s_2 \ge 0$,
 \[
  I(a, r_1, s_1) \cap I(a, r_2, s_2) = I(a, \max\{r_1, r_2\}, \min\{s_1,s_2\}). % todo das in dem theorem schreiben?
 \]
 Hence, we can suppose in the intersections from Theorem~\ref{thm:aperiodic_union_Fakk} that
 all letters are different.
 Then, with $a_i \ne a_j$ for $i,j \in \{1,\ldots,n\}$ distinct, we have
 \[ % todo buchstaben alle verschieden, in theorem oben schreiben...
  \bigcap_{i=1}^n I(a_i, r_i, s_i) = \bigshuffle_{\substack{i \in \{1,\ldots,n\} \\ s_i = \infty}} a_i^{r_i}
   \shuffle \bigshuffle_{\substack{i \in \{1,\ldots,n\} \\ s_i < \infty}} \{ a_i^{r_i}, a_i^{r_i+1}, \ldots, a_i^{s_i} \}
   \shuffle \Gamma^*
 \]
 with $\Gamma = \Sigma \setminus \{ a_1, \ldots, a_n \} \cup \{ a_i : \exists i \in \{1,\ldots,n\} : s_i = \infty \}$.
 As, for $u \in \Sigma^*$, we have $\perm(u) = \bigshuffle_{a \in \Sigma} a^{|u|_a}$
 and, by Theorem~\ref{thm:shuffle_properties}, the shuffle operation distributes over union,
 we can write the above set as a finite union
 of sets of the form $\perm(u) \shuffle \Gamma^*$.
 Then, Theorem~\ref{thm:aperiodic_union_Fakk} gives the claim.~\qed 
\end{proof}

\begin{remark}\label{rem:comm_aperiodic_aut}
 By a result from~\cite[Page 9]{Ito04},
 it follows that a letter which permutes 
 with every other letter has to permute the states
 of every strongly connected component. 
 This could be used to prove that the minimal automaton of an aperiodic commutative
 language cannot have non-trivial loops, i.e.,
 every loop must be a self-loop, which could also be used to give
 a proof of Proposition~\ref{prop:aperiodic_lang_form}.
\end{remark}

With Theorem~\ref{thm:finite_lang_it_shuffle}
we get the next result.

\begin{propositionrep}
 Let $u \in \Sigma^*$ and $\Gamma \subseteq \Sigma$.
 The iterated shuffle of $\perm(u) \shuffle \Gamma^*$
 is regular if and only if there exists $a \in \Sigma$ such that 
 $u \subseteq a^+$ or when $u \in \Gamma^*$.
\end{propositionrep}
\begin{proof}
 By Theorem~\ref{thm:shuffle_properties} and Theorem~\ref{thm:perm_semiring_hom},
 \begin{multline*}
  (\perm(u) \shuffle \Gamma^*)^{\itshuffle} \\
   = \{\varepsilon\} \cup  \perm(u^+) \shuffle \bigshuffle_{a \in \Gamma} a^* 
   = \{\varepsilon\} \cup \perm(u) \shuffle \perm( \{ u \} \cup \Gamma )^{\itshuffle}.
 \end{multline*}
 Then, apply Corollary~\ref{cor:finite_lang_it_shuffle}
 and the simple fact that a language $L \subseteq \Sigma^*$ is regular
 if and only if $L \cup \{\varepsilon\}$ is regular.~\qed % für beliebige disjunkte menge B, L \cup B vs L regularity
\end{proof}

Next, we give a simple sufficient criterion of regularity for a binary alphabet.

\begin{lemmarep}
 Let $\Sigma = \{a,b\}$ and $L \subseteq \Sigma^*$
 be regular. Then, if there exists $u \in \Sigma^*$
 such that $\perm(u) \shuffle \Sigma^* \subseteq \perm(L)$,
 then $\perm(L)$ is regular.
\end{lemmarep}
\begin{proof}
 Let $u \in \Sigma^*$ such that $\perm(u)\shuffle \Sigma^* \subseteq \perm(L)$.
 Note that $\perm(u) \shuffle \Sigma^*$ is regular.
 It is well-known~\cite{FerPSV2017} % todo parikh original paper zitieren?
 that, if $L \subseteq \Sigma^*$ is regular, then $\psi(\perm(L)) = \psi(L)$
 is a finite union of linear sets, i.e., sets of the
 form
 \[
  \left\{ v_0 + \sum_{i=1}^n k_i v_i \mid \{ k_1, \ldots, k_n \} \subseteq \mathbb N_0 \right\}
 \]
 for vectors $\{ v_0, v_1, \ldots, v_n \} \subseteq \mathbb N_0^{|\Sigma|}$. 
 Hence, we show that
 if $A \subseteq \mathbb N_0^{|\Sigma|}$
 is a linear subset, then $\psi^{-1}(A) \cup \perm(u) \shuffle \Sigma^*$
 is regular, which, inductively, gives our claim.
 Write
 \[
  A = \left\{ v_0 + \sum_{i=1}^n k_i v_i \mid \{ k_1, \ldots, k_n \} \subseteq \mathbb N_0 \right\}
 \]
 for vectors $\{ v_0, v_1, \ldots, v_n \} \subseteq \mathbb N_0^{|\Sigma|}$. Without loss of generality,
 we suppose none of the vectors $v_1, \ldots, v_n$ is zero, the vectors 
 $v_1, \ldots, v_m$ are axis-parallel, i.e., exactly one entry is non-zero,
 and the vectors $v_{m+1}, \ldots, v_n$ are sloped, i.e., we have at least two non-zero entries.
 Choosing words $u_i \in \Sigma^*$, $i \in \{0,\ldots,n\}$, with $\psi(u_i) = v_i$,
 we have
 \begin{equation}\label{eqn:form_lang}
  \psi^{-1}(A) = \perm(u_0) \shuffle \bigshuffle_{i=1}^n \perm(u_i^*).
 \end{equation}
 If $m = n$, then all the vectors are axis-parallel.
 Then, the words $u_1,\ldots,u_n$ are unary and if we write $u_i \in a_i^*$
 in this case, we find $\perm(u_i^*) = u_i^*$ and 
 $\psi^{-1}(A)$ is regular. Hence, $\psi^{-1}(A) \cup \perm(u) \shuffle \Sigma^*$
 is regular.
 So, suppose $m < n$.
 As $\Sigma = \{a,b\}$, for any $i \in \{m+1,\ldots,n\}$,
 in the the vector $v_i$ all entries are non-zero.
 Hence, for the fixed $u \in \Sigma^*$ chosen above, we can choose numbers $N_i \ge 0$ such that
 \[
  \psi(u) \le v_0 + N_i v_i.
 \]
 So, if $v = v_0 + k_1 v_1 + \ldots + k_n v_n$
 with $k_i \ge N_i$ for $i \in \{1,\ldots,n\}$,
 then $\psi^{-1}(v) \subseteq \perm(u) \shuffle \Sigma^*$
 and $\perm(u) \shuffle \Sigma^* \cup \psi^{-1}(A)$ equals
 \begin{multline*}
 \perm(u) \shuffle\Sigma^* \cup \\
   \psi^{-1}\left( \left\{ v_0 + \sum_{i=1}^n k_i v_i \mid \{ k_1, \ldots, k_n \} \subseteq \{0,\ldots, \max\{N_1, \ldots,N_n\} - 1 \} \right\} \right). 
 \end{multline*}
 Hence, it is a regular language.~\qed
\end{proof}

Lastly, a few examples of aperiodic commutative languages, some
of them yielding non-regular languages and some of them regular languages
when applying the iterated shuffle.

\begin{example} Let $\Sigma = \{a,b,c\}$.

\begin{enumerate}
\item The iterated shuffle of $\{ab,ba\} \cup \{c\} \shuffle \{a,b\}^*$ 
 is not regular.

\item The iterated shuffle of $\{ab,ba\} \shuffle \{c\}^* \cup \{ac\}\shuffle \{a,b\}^*$
 is not regular.
 
\item The iterated shuffle of $\{ ab,ba \} \cup \{c\}\shuffle \{a,b\}^* \cup \perm(abb) \shuffle \{a,b\}^*$
 is regular.
 
\item The iterated shuffle of $\{ab,ba\} \cup \{c\}\shuffle \{a,b\}^* \cup \perm(abb) \shuffle \{a\}^* \cup \{bb\}$
 is regular.
\end{enumerate}
\end{example}

\subsection{Decision Procedures} % dfa membership L-complete?? steht in M.Holzer, Desc compl survey...
\label{subsec:dec_procedures}

In~\cite{DBLP:journals/tcs/Gohon85,GinsburgSpanier66} it was shown that for regular $L \subseteq \Sigma^*$,
it is decidable if $\perm(L)$ is regular. As $\perm(L)^{\itshuffle} = \perm(L^*)$, also
the regularity of the iterated shuffle on commutative regular languages is decidable. This result
was also shown directly, without citing~\cite{DBLP:journals/tcs/Gohon85,GinsburgSpanier66},
in~\cite{Ito04,DBLP:conf/dmtcs/ImrehIK96}.
However, the precise computational complexity was not clear, and by a statement
given in~\cite[Theorem 45]{FerPSV2017} it follows that for a regular language
given by a regular expression it is $\NP$-hard to decide if the commutative closure
is regular.
On the contrary, the conditions stated in Theorem~\ref{thm:finite_lang_it_shuffle} 
 could be tested in polynomial % todo diese begriffe einführen.
time for a finite commutative language given by a deterministic, a non-deterministic or a 
regular expression as input. 
This follows as non-emptiness of intersection 
with the fixed languages $\Sigma^*a\Sigma^*$ and $a^+$, $a \in \Sigma$, could be done in polynomial
time by the product automaton construction.
% This follows as (1) the membership problem is solvable in polynomial
% time for each of them, even in $\NL$ for automata as input, (2) non-emptiness of intersection 
% with the fixed languages $\Sigma^*a\Sigma^*$, $a \in \Sigma$, could be done in polynomial
% time by the product automaton construction, \todo{komma vor when?}
% and (3) the condition when a letter $a\in \Sigma$ is contained in $a \in \Gamma_1 \cup \ldots \cup \Gamma_n$ in the statement of 
% Theorem~\ref{thm:aperiodic_it_shuffle} is true, for a given automaton with $n$
% states, precisely if the recognized language has non-empty intersection with $\Sigma^*a^n\Sigma^*$,
% which could also be checked in polynomial time as before.

%For example, if it is given as a non-deterministic automaton,
%the conditions $\Sigma^*a \Sigma^*$ could be checked in time

\section{Conclusion}

A general criterion as given for finite (commutative) languages
in Theorem~\ref{thm:finite_lang_it_shuffle}, which gives a polynomial time decision procedure,
for general commutative regular languages is an open problem.
% lower bound result für diesen shuffle, a^{p-1}(a^p) liefert auch (n/k)^k todo
% + index/period bound für projection
For the subclass closed under iterated shuffle identified
in Subsection~\ref{subsec:it_shuffle}, a sharp bound for the size of a recognizing automaton of the iterated shuffle is unknown.

\smallskip \noindent \footnotesize
\textbf{Acknowledgement.} I thank the anonymous reviewers for careful reading, pointing out typos and unclear formulations
and providing additional references.

\bibliographystyle{splncs04}
\bibliography{ms}
\end{document}